\documentclass[11pt,onecolumn]{IEEEtran}
\usepackage{setspace}
\usepackage{hyperref}
\usepackage{graphicx}
\usepackage{amssymb}
\usepackage{amsmath}%
\usepackage{amsfonts}%
\usepackage{amssymb}%
\usepackage{amsthm}
\usepackage{graphicx}
\usepackage{subfigure}
\usepackage{mathrsfs}
\usepackage[usenames,dvipsnames]{xcolor}
\usepackage{mdframed}
\usepackage{bm}

\newtheorem{thm}{Theorem}
\newtheorem{lem}{Lemma}

\newtheorem{conj}{Conjecture}
\newtheorem{cor}{Corollary}

\theoremstyle{definition}

\newtheorem{defn}{Definition}

\newtheorem{remark}{Remark}

\newcommand{\etal}{\textit{et al.}~}

\newcommand{\calX}{\mathcal{X}}
\newcommand{\calA}{\mathcal{A}}

\newcommand{\calY}{\mathcal{Y}}

\newcommand{\calC}{\mathcal{C}}
\newcommand{\calS}{\mathcal{S}}

\newcommand{\calH}{\mathcal{H}}
\newcommand{\calP}{\mathcal{P}}

\newcommand{\bD}{\mathbf{D}}
\newcommand{\bx}{\mathbf{x}}

\renewcommand{\tilde}{\widetilde}

\newcommand{\Reals}{\mathbb{R}}

\newcommand{\normX}[2]{\| #1 \|_{#2}}
\newcommand{\normEuc}[1]{\| #1 \|_2}

\newcommand{\ones}{\mathbf{1}}

\newcommand{\defined}{\triangleq}

\newcommand{\diag}[1]{\mathrm{diag}\left( #1 \right)}
\newcommand{\pxy}{p_{X,Y}}
\newcommand{\px}{p_X}
\newcommand{\py}{p_Y}

\newcommand{\ExpVal}[2]{\mathbb{E}\left[ #2 \right]}

\newcommand{\Pygx}{\mathbf{P}_{Y|X}}

\newcommand{\Px}{\mathbf{p}_X}

\newcommand{\Py}{\mathbf{p}_Y}

\newcommand{\Yt}{\tilde{Y}}

\newcommand{\by}{\mathbf{y}}

\newcommand{\bLambda}{\pmb{\Lambda}}

\newcommand{\EE}[1]{\ExpVal{}{#1}}
\newcommand{\whB}{\widehat{B}}

\newcommand{\xb}{\mathbf{x}}
\newcommand{\yb}{\mathbf{y}}
\newcommand{\fb}{\mathbf{f}}
\newcommand{\gb}{\mathbf{g}}
\newcommand{\bP}{\mathbf{P}}
\newcommand{\eye}{\mathbf{I}}
\newcommand{\bQ}{\mathbf{Q}}
\newcommand{\bU}{\mathbf{U}}
\newcommand{\bSigma}{\mathbf{\Sigma}}
\newcommand{\bV}{\mathbf{V}}
\newcommand{\bT}{\mathbf{T}}
\newcommand{\bbH}{\mathbf{H}}

\newcommand\independent{\protect\mathpalette{\protect\independenT}{\perp}}
\def\independenT#1#2{\mathrel{\rlap{$#1#2$}\mkern2mu{#1#2}}}

\let\emptyset\varnothing

\definecolor{light-gray}{gray}{.97}

\newmdtheoremenv[%
  backgroundcolor=light-gray, %
  linecolor=black,
  linewidth =1pt,%
  skipabove = 10pt,%
  skipbelow = 10pt
]{work}{Work in progress}

\newmdtheoremenv[%
  fontcolor=BrickRed,
  linecolor=black,
  linewidth =1pt,%
  skipabove = 10pt,%
  skipbelow = 10pt
]{todo}{To Do}

\newmdtheoremenv[%
  fontcolor=BrickRed,
  linecolor=black,
  linewidth =1pt,%
  skipabove = 10pt,%
  skipbelow = 10pt
]{disc}{Disclaimer}

\title{An Exploration of the Role of Principal Inertia Components in Information
Theory}
\author{Flavio P. Calmon, Mayank Varia, Muriel M\'edard
\thanks{F.~P.~Calmon and M.~M\'edard are with the Research Laboratory
of Electronics at the Massachusetts Institute of Technology, Cambridge, MA (email:
flavio@mit.edu; medard@mit.edu).

M.~Varia is with the MIT Lincoln Laboratory, Lexington, MA (e-mail:
mayank.varia@ll.mit.edu).

This work is sponsored by the Intelligence Advanced Research Projects
Activity under Air Force Contract FA8721-05-C-002. Opinions,
interpretations, conclusions and recommendations are those of the authors
and are not necessarily endorsed by the United States Government.
}
}

\begin{document}
\maketitle

\begin{abstract}
    The principal inertia components of the joint distribution of two random
    variables $X$ and $Y$ are inherently connected to how an observation of $Y$
    is statistically related to a hidden variable
    $X$. In this paper, we explore this connection within an information theoretic
    framework. We show that, under certain symmetry conditions, the principal
    inertia components play an important role in estimating one-bit functions of
    $X$, namely $f(X)$, given an observation of $Y$. In particular, the
    principal inertia components bear an interpretation as filter coefficients
    in the linear transformation of $p_{f(X)|X}$ into $p_{f(X)|Y}$. This
    interpretation naturally leads to the conjecture that the mutual information
    between $f(X)$ and $Y$ is maximized when all the principal inertia
    components have equal value. We also study the role of the principal inertia
    components in the Markov chain $B\rightarrow X\rightarrow Y\rightarrow
    \whB$, where $B$ and $\whB$ are binary random variables. We illustrate our
    results for the setting where  $X$ and $Y$ are binary strings
    and $Y$ is the result of sending $X$ through an additive noise binary
    channel.
  \end{abstract}


\section{Introduction}

Let $X$ and $Y$ be two discrete random variables with finite support $\calX$ and
$\calY$, respectively.  $X$ and $Y$ are related through a conditional
distribution (channel), denoted by $p_{Y|X}$. For each $x\in \calX$,
$p_{Y|X}(\cdot|x)$ will be a vector on the $|\calY|$-dimensional simplex, and
the position of these vectors on the simplex will determine the nature of the
relationship between $X$ and $Y$.  If $p_{Y|X}$ is fixed, what can be
learned about $X$ given an observation of $Y$, or the degree of accuracy of what
can be inferred about $X$ \textit{a posteriori}, will  then depend on the
marginal distribution $p_X$. The value $p_X(x)$, in
turn, ponderates the corresponding vector $p_{Y|X}(\cdot|x)$ akin to a mass.  
As a simple example, if $|\calX|=|\calY|$ and the vectors $p_{Y|X}(\cdot|x)$ are located
on distinct corners of the simplex, then $X$ can be perfectly learned from $Y$.
As another example, assume that the vectors $p_{Y|X}(\cdot|x)$ can be grouped into two
clusters located near opposite corners of the simplex. If the sum of
the masses induced by $p_X$ for each cluster is approximately $1/2$, then one
may expect to reliably infer on the order of 1 unbiased bit of $X$ from an
observation of $Y$.

The above discussion naturally leads to considering the use of techniques
borrowed from classical mechanics.  For a given inertial frame of reference, the
mechanical properties of a collection of distributed point masses can be
characterized by the moments of inertia of the system. The moments of inertia
measure how the weight of the point masses is distributed around the center of
mass. An analogous metric exists for the distribution of the vectors $p_{Y|X}$
and masses $p_X$ in the simplex, and it is the subject of study of a branch of
applied statistics called \textit{correspondence analysis}
(\cite{greenacre_theory_1984,greenacre_geometric_1987}). In correspondence
analysis, the joint distribution  $p_{X,Y}$  is decomposed in terms of the
\textit{principal inertia components}, which, in some sense, are analogous to
the moments of inertia of a collection of point masses. In mathematical
probability, the study of principal inertia components dates back to Hirschfeld
\cite{hirschfeld_connection_1935}, Gebelein \cite{gebelein_statistische_1941},
Sarmanov \cite{sarmanov1962maximum} and R\'enyi \cite{renyi_measures_1959}, and
similar analysis have also recurrently appeared in the information theory and
applied probability literature. We present the formal definition of principal
inertia components and a short review of the relevant literature in the next
section\footnote{We encourage the readers that are unfamiliar with the topic to
  skip ahead and read Section \ref{sec:notation} and then return to this
introduction.}.


The distribution of the vectors $p_{Y|X}$ in the simplex or, equivalently, the
principal inertia components of the joint distribution of $X$ and $Y$,
is inherently connected to how an observation of $Y$ is statically related to $X$.   
In this paper, we explore
this connection within an information theoretic framework. We show that, under
certain assumptions, the principal inertia components play an important part in
estimating a one-bit function of $X$, namely $f(X)$ where $f:\calX\rightarrow
\{0,1\}$, given an observation of $Y$: they can be understood  as the filter
coefficients in the linear transformation of $p_{f(X)|X}$ into $p_{f(X)|Y}$.
Alternatively, the principal inertia components can bear an interpretation as
noise, in particular when $X$ and $Y$ are binary strings.
We also show that maximizing the principal inertia components is equivalent to
maximizing the first-order term of the Taylor series expansion of certain convex
measures of information between $f(X)$ and $Y$. We conjecture that, for symmetric
distributions of $X$ and $Y$ and a given upper bound on the value of the
largest principal inertia component, $I(f(X);Y)$ is maximized when all the
principal inertia components have the same value as the largest principal
inertia component. This is equivalent to $Y$ being the result of passing $X$
through a $q$-ary  symmetric channel. This conjecture, if proven, would imply
that the conjecture made by Kumar and Courtade in \cite{kumar_which_2013}.

Finally, we study the Markov chain $B\rightarrow X \rightarrow Y \rightarrow
\whB$, where $B$ and $\whB$ are binary random variables, and the role of the
principal inertia components in characterizing the relation between $B$ and
$\whB$. We show that that this relation is linked to solving a non-linear
maximization problem, which, in turn, can be solved when $\whB$ is an unbiased
estimate of $B$, the joint distribution of $X$ and $Y$ is symmetric and
$\Pr\{B=\whB=0\}\geq \EE{B}^2$. We illustrate this result for the setting where $X$ is a
binary string and $Y$ is the result of sending $X$ through a memoryless binary
symmetric channel. We note that this is a similar setting to the one considered
by Anantharam \etal in \cite{anantharam_hypercontractivity_2013}.

The rest of the paper is organized as follows.  Section \ref{sec:notation}
presents the notation and definitions used in this paper, and discusses some of
the related literature. Section \ref{sec:conf} introduces the notion of conforming
distributions and ancillary results. Section \ref{sec:functions} presents
results concerning the role of the principal inertia components in inferring
one-bit functions of $X$ from an observation of $Y$, as well as the linear
transformation of $p_X$ into $p_Y$ in certain symmetric settings. We argue
that, in such settings, the principal inertia components can be viewed as filter
coefficients in a linear transformation. In particular, results for binary
channels with additive noise are derived using techniques inspired by 
Fourier analysis of Boolean functions. Furthermore, Section \ref{sec:functions}
also introduces a conjecture that encompasses the one made by Kumar and Courtade
in \cite{kumar_which_2013}. Finally, Section \ref{sec:estimators} provides
further evidence for this conjecture by investigating the  Markov chain
$B\rightarrow X\rightarrow Y \rightarrow \whB$ where $B$ and $\whB$ are binary
random variables.

\section{Principal Inertia Components}
\label{sec:notation}
\subsection{Notation}
\label{sec:notA}

We denote matrices by bold capitalized letters (e.g. $\mathbf{A}$) and vectors
by bold lower case letters (e.g. $\xb$). The $i$-th component of a vector $\xb$
is denoted by $\xb_i$. Random variables are denoted by upper-case letters (e.g.
$X$ and $Y$). We define $[n]\defined \{1,\dots,n\}$.

Throughout the text we assume that $X$ and $Y$ are discrete random variables
with  finite support sets $\calX$ and $\calY$. Unless otherwise specified, we
let, without loss of generality, $\calX=[m]$ and $\calY=[n]$. The joint
distribution matrix of $\bP$ is an $m\times n$ matrix with $(i,j)$-th entry
equal to $p_{X,Y}(i,j)$. We denote by $\Px$ (respectively, $\Py$) the vector
with $i$-th entry equal to $p_X(i)$ (resp. $p_Y(i)$).
$\mathbf{D}_X=\diag{\Px}$ and $\mathbf{D}_Y=\diag{\Py}$ are matrices with diagonal
entries equal to $\Px$ and $\Py$, respectively, and all other entries equal to
0. The matrix $\Pygx\in \Reals^{m\times n}$ denotes the matrix with $(i,j)$-th
entry equal to $p_{Y|X}(j|i)$. Note that  $\bP=\bD_X\Pygx$.

For a given joint distribution matrix $\bP$, the set of all vectors contained in
the unit cube in $\Reals^n$ that satisfy $\normX{\bP\bx}{1}=a $ is given by
\begin{equation}
  \label{eq:Cdef}
\calC^n(a,\bP)\defined \{\xb\in \Reals^n|0\leq \xb_i \leq 1, \normX{\bP\bx}{1}=a\}.
\end{equation}
The set of all $m\times n$ probability distribution matrices is given by
$\calP_{m,n}$. 

For $x^n\in \{-1,1\}^n$ and
$\calS\subseteq [n]$, $\chi_\calS(x^n)\defined\prod_{i\in \calS}x_i$ (we
consider $\chi_\emptyset(x) = 1$). For $y^n\in \{-1,1\}^n$, $a^n=x^n\oplus y^n$
is the vector resulting from the entrywise product of $x^n$ and $y^n$, i.e. $a_i
= x_iy_i$, $i\in [n]$.

Given two probability distributions $p_{X}$ and $q_X$ and $f(t)$ a smooth convex function
defined for $t>0$ with $f(1)=0$, the $f$-divergence is
defined as \cite{_information_2004} 
    \begin{equation}
        D_f(p_X||q_X) \defined \sum_x q_{X}(x) f \left( \frac{p_X(x)}{q_X(x)}
        \right).
    \end{equation}
The $f$-information is given by
    \begin{equation}
        I_f(X;Y)\defined D_f(p_{X,Y}||p_Xp_Y).
    \end{equation}
When $f(x)=x\log(x)$, then $I_f(X;Y)=I(X;Y)$. A study of information metrics
related to $f$-information  was given in \cite{polyanskiy2010arimoto} in the
context of channel coding converses.

\subsection{Principal Inertia Components and Decomposing the Joint
  Distribution Matrix}
  \label{sec:PIC}

We briefly define in this section the \textit{principal inertia decomposition}
of the joint distribution matrix $\bP$. The term ``principal inertia'' is
borrowed from the correspondence analysis literature
\cite{greenacre_theory_1984}. The study of the principal
inertia components of the joint distribution of two random variables dates back to
Hirshfield \cite{hirschfeld_connection_1935}, Gebelein
\cite{gebelein_statistische_1941}, Sarmanov \cite{sarmanov1962maximum}  and R\'enyi \cite{renyi_measures_1959}, having appeared in the work of
Witsenhausen \cite{witsenhausen_sequences_1975}, Ahlswede and G\'acs
\cite{ahlswede_spreading_1976} and, more recently, Anantharam \etal
\cite{anantharam_maximal_2013}, Polyanskiy
\cite{polyanskiy_hypothesis_2012} and Calmon \etal \cite{calmon_bounds_2013},
among others.
For an overview, we refer the reader to
\cite{anantharam_maximal_2013,calmon_bounds_2013}.

\begin{defn}
  We call the singular value decomposition $\mathbf{D}_X^{-1/2}\bP
  \mathbf{D}_Y^{-1/2}=\mathbf{U\Sigma V}^T$ the \textit{principal inertia decomposition}
  of $X$ and $Y$, where $\bSigma$ is a diagonal matrix with $\diag{\bSigma}=(1,\sigma_1,\dots,\sigma_d)$ and
  $d=\min(m,n)-1$.  The values $\sigma_i^2$, $i=1,\dots,d$, 
are called the \textit{principal inertia components} of $X$ and $Y$. In particular
$\rho_m(X;Y)=\sigma_1$, where $\rho_m(X;Y)$ denotes the maximal correlation coefficient of $X$ and
$Y$. The maximal correlation coefficient, in turn, is given by
\begin{equation*}
  \rho_m(X;Y) \defined \sup \left\{ \EE{f(X)g(Y)}|
  \EE{f(X)}=\EE{g(Y)}=0,\EE{f(X)^2}=\EE{g(X)^2}=1 \right\}.
\end{equation*}
\end{defn} 
The values $\sigma_1,\dots,\sigma_d$ in the previous definition are the spectrum
of the conditional expectation operator $(Tf)(x)\defined \EE{f(Y)|X=x}$, where
$f:\calY\rightarrow \Reals$ \cite{renyi_measures_1959}. Indeed, the spectrum of
$T$ and the principal inertia components are entirely equivalent when $X$ and
$Y$ have finite support sets. Nevertheless, the reader should note
that the analysis based on the conditional expectation operator lends itself to
more general settings, including random variables with continuous support. We do
not pursue this matter further here, since our focus is on discrete random
variables with finite support.

The principal inertia components satisfy the data processing inequality (see, for example,
\cite{polyanskiy_hypothesis_2012,calmon_bounds_2013,kang_new_2011}): if $X
\rightarrow Y \rightarrow Z$ and $\sigma_i$ are the principal inertia components
of $X$ and $Y$ and $\tilde{\sigma_i}$ are the principal inertia components of $X$
and $Z$, then $\sum_{i=1}^k \tilde{\sigma_i}^2\leq\sum_{i=1}^k \sigma_i^2$ for all
$k$. Furthermore, for a fixed marginal distribution $p_X$, $\sum_{i=1}^k
\sigma_i^2$ is convex
in $p_{Y|X}$. Note the joint distribution matrix $\bP$ as can be written as 
\begin{align}
  \label{eq:PIDecomp}
  \bP = \mathbf{D}_X^{1/2}\mathbf{U\Sigma V}^T \mathbf{D}_Y^{1/2}.
\end{align}

\section{Conforming distributions}
\label{sec:conf}

In this paper we shall recurrently use probability distribution matrices that
are symmetric and positive-semidefinite. This motivates the following
definition.

\begin{defn}
  A joint distribution $p_{X,Y}$ is said to be \textit{conforming} if the
  corresponding matrix $\bP$ satisfies $\bP=\bP^T$ and $\bP$ is
  positive-semidefinite.
\end{defn}

\begin{remark}
If $X$ and $Y$ have a conforming joint distribution, then they have the same
marginal distribution.  Consequently, $\mathbf{D}\defined
\mathbf{D}_X=\mathbf{D}_Y$, and $\bP=\mathbf{D}^{1/2}\mathbf{U\Sigma U}^T
\mathbf{D}^{1/2}$. 
\end{remark}

Symmetric channels\footnote{We say that a channel is
  symmetric if $\calX=\calY=[m]$ and $p_{Y|X}(i|j)=p_{Y|X}(j|i)$ $\forall
i,j\in [m]$.} are closely related to conforming probability distributions. 
We shall illustrate this relation in the next lemma and in Section
\ref{sec:functions}.

\begin{lem}
  \label{lem:conf}
    If $\bP$ is conforming, then the corresponding conditional distribution
    matrix $\Pygx$ is positive semi-definite. Furthermore, for any symmetric channel
    $\Pygx=\Pygx^T$, there is an input distribution
    $\Px$ (namely, the uniform distribution) such that the 
    principal inertia components of $\bP=\bD_X\Pygx$ correspond to the square of the
    eigenvalues of $\Pygx$. In this case, if $\Pygx$ is also
    positive-semidefinite, then $\bP$ is conforming.
\end{lem}
\begin{proof}
  Let $\bP$ be conforming and $\calX=\calY=[m]$. Then $\Pygx = \mathbf{D}^{-1/2}\mathbf{U\Sigma U}^T
    \mathbf{D}^{1/2}=\mathbf{Q}\bSigma\mathbf{Q}^{-1}$, where
    $\mathbf{Q}=\bD^{-1/2}\bU$. It follows that $\diag{\bSigma}$ are the
    eigenvalues of $\Pygx$, and, consequently, $\Pygx$ is positive
    semi-definite.

    Now let $\Pygx=\Pygx^T=\bU\bLambda\bU^T$. The entries of $\bLambda$ here are
    the eigenvalues of $\Pygx$ and not necessarily positive. Since $\Pygx$ is
    symmetric, it is also doubly stochastic, and for $X$ uniformly distributed
    $Y$ is also uniformly distributed. Therefore, $\bP$ is symmetric, and
    $\bP=\bU\bLambda\bU^T/m$. It follows directly that the principal inertia
    components of $\bP$ are exactly the diagonal entries of $\bLambda^2$, and if
    $P_{Y|X}$ is positive-semidefinite then $\bP$ is conforming.
\end{proof}

The $q$-ary symmetric channel, defined below, is of particular interest to some
of the results derived in the following sections. 

\begin{defn}
  The $q$-ary symmetric channel with crossover probability $\epsilon\leq1-q^{-1}$, also
    denoted as  $(\epsilon,q)$-SC, is
    defined as the channel with input $X$ and output $Y$ where
    $\calX=\calY=[q]$ and
    \begin{align*}
        p_{Y|X}(y|x) = 
            \begin{cases}
                1-\epsilon& \mbox{if } x=y\\
                 \frac{\epsilon}{q-1}& \mbox{if } x\neq y.
            \end{cases}
    \end{align*}
\end{defn}

Let $X$ and $Y$ have a conforming joint distribution matrix with
$\calX=\calY=[q]$ and principal inertia components
$\sigma_1^2,\dots,\sigma_{d}^2$. The following  lemma shows that conforming $\bP$ can be
transformed into the joint distribution of a $q$-ary symmetric channel with input
distribution $p_X$ by setting $\sigma_1^2=\sigma_2^2=\dots=\sigma_d^2$, i.e. making
all principal inertia components equal to the largest one.

\begin{lem}
  \label{lem:qary}
  Let $\bP$ be a conforming joint distribution matrix of $X$ and
  $Y$, with $X$ and $Y$ uniformly distributed, $\calX=\calY=[q]$,
  $\bP=q^{-1}\mathbf{U\Sigma U}^T$ and
  $\mathbf{\Sigma}=\diag{1,\sigma_1,\dots,\sigma_d}$. For
  $\tilde{\mathbf{\Sigma}}=\diag{1,\sigma_1,\dots,\sigma_1}$, let $X$ and
  $\tilde{Y}$ have joint distribution $\tilde{\bP}=q^{-1}\mathbf{U\widetilde{\Sigma}
U}^T $.  Then,  $\tilde{Y}$ is the result of passing $X$ through a
$(\epsilon,q)$-SC, with
  \begin{align}
    \label{eq:epsilon}
    \epsilon = \frac{(q-1)(1-\rho_m(X;Y))}{q}.
  \end{align}
\end{lem}

\begin{proof}
  The first column of $\mathbf{U}$ is $\Px^{1/2}$ and, since $X$ is
  uniformly distributed, $\Px^{1/2}=q^{-1/2}\ones$. Therefore
\begin{align}
    \tilde{\bP}&=q^{-1}\mathbf{U\tilde{\Sigma} U}^T \nonumber \\
    &=q^{-1} \sigma_1 \eye +q^{-2}(1-\sigma_1)\ones\ones^T. \label{eq:qarydist}
\end{align}    
Consequently, $\tilde{\bP}$ has diagonal entries equal to 
$(1+(q-1)\sigma_1)/q^2$ and all other entries equal to $(1-\sigma_1)/q^2$. The
result follows by noting that  $\sigma_1=\rho_m(X;Y)$. 
\end{proof}

\begin{remark}
  For $X$, $Y$ and $\tilde{Y}$ given in the previous lemma, a natural question
  that arises is whether $Y$ is a degraded version of $\Yt$, i.e. $X\rightarrow
  \tilde{Y}\rightarrow Y$. Unfortunately, this is \textbf{not true} in general,
  since the matrix $\mathbf{U\widetilde{\Sigma}^{-1}\Sigma U^T}$ does not
  necessarily contain only positive entries, although it is doubly-stochastic.
  However, since the principal inertia components of $X$ and $\tilde{Y}$ upper
  bound the principal inertia components of $X$ and $Y$, it is natural to expect
  that, at least in some sense, $\tilde{Y}$ is more informative about $X$ than
  $Y$. This intuition is indeed correct for certain estimation problems where a
  one-bit function of $X$ is to be inferred from a single observation $Y$ or
  $\tilde{Y}$, and will be investigated in the next section.
\end{remark}


\section{One-bit Functions and Channel Transformations}
\label{sec:functions}

Let $B\rightarrow X \rightarrow Y$, where $B$ is a binary random variable. When
$X$ and $Y$ have a conforming probability distribution, the principal inertia components of $X$ and
$Y$ have a particularly interesting interpretation: they can be understood as
the filter coefficients in the linear transformation of $p_{B|X}$ into
$p_{B|Y}$. In order to see why this is the case, consider the joint
distribution of $B$ and $Y$, denoted here by $\bQ$, given by
\begin{align}
  \label{eq:OBF}
  \bQ = [\fb~ ~1-\fb]^T\bP= [\fb~ ~1-\fb]^T\bP_{X|Y}\bD_Y=[\gb~ ~ 1-\gb]^T\bD_Y,
\end{align}
where $\fb\in \Reals^m$ and $\gb\in \Reals^n$ are column-vectors with
$\fb_i=p_{B|X}(0|i)$ and $\gb_j=p_{B|Y}(0|j)$. In particular, if $B$ is a
deterministic function of $X$, $\fb\in \{0,1\}^m$.

If $\bP$ is conforming and $\calX=\calY=[m]$, then $\bP=\bD^{1/2}\bU\bSigma\bU^T\bD^{1/2}$,
where $\bD=\bD_X=\bD_Y$. Assuming $\bD$ fixed, the joint distribution $\bQ$ is entirely specified by the
linear transformation of $\fb$ into $\gb$. Denoting $\bT\defined \bU^T\bD^{1/2}$, this
transformation is done in three steps:
\begin{enumerate}
  \item (Linear transform) $\widehat{\fb}\defined \bT\fb$,
  \item (Filter) $\widehat{\gb}\defined \bSigma\widehat{\fb}$, where the diagonal of
    $\bSigma^2$ are the principal inertia components of $X$ and $Y$, 
  \item (Inverse transform)  $\gb= \bT^{-1}\widehat{\gb}$.
\end{enumerate}
Note that $\widehat{\fb}_1=\widehat{\gb}_1=1-\EE{B}$ and  $\widehat{\gb}=
\bT\gb$. Consequently, the principal inertia coefficients of $X$ and $Y$ bear an interpretation
as the filter coefficients in the linear transformation of  $p_{B|X}(0|\cdot)$
into $p_{B|Y}(0|\cdot)$.




A similar interpretation can be made for symmetric channels, where
$\Pygx=\Pygx^T=\bU\bLambda\bU^T$ and $\Pygx$ acts as the matrix of the linear
transformation of $\Px$ into
$\Py$.  Note that $\Py = \Pygx\Px$, and,
consequently, $\Px$ is transformed into $\Py$ in the same three steps as before:
\begin{enumerate}
  \item (Linear transform) $\widehat{\Px}= \bU^T\Px$,
  \item (Filter) $\widehat{\Py}=\bLambda\widehat{\Px}$, where the diagonal of
    $\bLambda^2$ are the principal inertia components of $X$ and $Y$ in the
    particular case when $X$ is uniformly distributed (Lemma \ref{lem:conf}), 
  \item (Inverse transform)  $\Py= \bU\widehat{\Py}$.
\end{enumerate}
From this perspective, the vector $\mathbf{z}=\bU\bLambda\ones m^{-1/2}$ can be
understood as a proxy for the ``noise effect'' of the channel. Note that
$\sum_i\mathbf{z}_i=1$. However, the entries of $\mathbf{z}$ are not necessarily
positive, and $\mathbf{z}$ might not be a \textit{de facto} probability
distribution.

We now illustrate these ideas by investigating binary channels with additive
noise in the next section, where $\bT$ will correspond to the
well-known Walsh-Hadamard transform matrix.  

\subsection{Example: Binary Additive Noise Channels}
\label{sec:binAdd}

In this example, let $\calX^n, \calY^n\subseteq \{-1,1\}^n$ be the support sets of $X^n$ and
$Y^n$, respectively.  We define two sets of channels that transform
$X^n$ into $Y^n$. In each set definition, we  assume the conditions for
$p_{Y^n|X^n}$ to be a valid probability distribution (i.e. non-negativity and unit
sum).

 \begin{defn}
    The set of \textit{parity-changing channels} of block-length $n$, denoted
    by $\calA_{n}$, is defined as:
    \begin{align}
      \calA_{n} \defined \left\{
        p_{Y^n|X^n} \mid \forall
        \calS\subseteq[n],~\exists c_{\calS} \in
        [-1,1]  \mbox{ s.t. }
        \EE{\chi_\calS(Y^n)|X^n}=c_{\calS}\chi_\calS(X^n)  \right\}. \label{eq:PAchannels}
    \end{align}
    The set of all \textit{binary additive noise channels} is given by
        \begin{align*}
          \mathcal{B}_{n} \defined \left\{
        p_{Y^n|X^n} \mid  \exists Z^n \mbox{ s.t. }
        Y^n=X^n\oplus Z^n, \mbox{ supp}(Z^n)\subseteq\{-1,1\}^n, Z^n\independent X^n \right\}.
    \end{align*}
 \end{defn}
 The definition of parity-changing channels is inspired by results from the
 literature on Fourier analysis of Boolean functions. For an overview of the
 topic, we refer the reader to the survey \cite{odonnell_topics_2008}.  The set
 of binary additive noise channels, in turn, is widely used in the information
 theory literature. The following theorem shows that both characterizations are
 equivalent.
%
%
%

 \begin{thm}
   \label{thm:AB}
   $\calA_n=\mathcal{B}_n$.
 \end{thm}
 \begin{proof}
     Let $Y^n = X^n \oplus Z^n$ for some
     $Z^n$ distributed over $\{-1,1\}^n$ and independent of $X^n$. Thus
        \begin{align*}
          \EE{\chi_\calS(Y^n)|X^n}&=\EE{\chi_\calS(Z^n\oplus X^n) \mid X^n}\\
          &=\EE{\chi_\calS(X^n)\chi_\calS(Z^n) \mid X^n}\\
          &=\chi_\calS(X^n) \EE{\chi_\calS(Z^n)},
        \end{align*}
     where the last equality follows from the assumption that $X^n\independent Z^n$. By
     letting $c_\calS =  \EE{\chi_\calS(Z^n)}$, it follows that $p_{Y^n|X^n}\in
     \calA_n$ and, consequently, $\mathcal{B}_n\subseteq \calA_n$.

     Now let $y_n$ be fixed and $\delta_{y^n}:\{-1,1\}^n\rightarrow \{0,1\}$ be given by 
        \begin{align*}
          \delta_{y^n}(x^n)=
            \begin{cases}
                1,& x^n=y^n,\\
                0,& \mbox{otherwise.}
            \end{cases}
        \end{align*}
        Since the function $\delta_{y^n}$ has Boolean inputs, it  can be expressed
        in terms of its Fourier expansion  \cite[Prop. 1.1]{odonnell_topics_2008} as
       \begin{equation*}
         \delta_{y^n}(x^n) = \sum_{\calS\subseteq [n]} \widehat{d}_\calS
         \chi_\calS(x^n).
       \end{equation*}
       Now let $p_{Y^n|X^n}\in \calA_n$. Observe that $p_{Y^n|X^n}(y^n|x^n) =
       \EE{\delta_{y^n}(Y^n)\mid X^n=x^n}$ and, for $z^n\in \{-1,1\}^n$,
     \begin{align*}
       p_{Y^n|X^n}(y^n\oplus z^n|x^n\oplus z^n) &= \EE{\delta_{y^n\oplus z^n}(Y^n)\mid
     X^n=x^n\oplus z^n}\\
     & =\EE{\delta_{y^n}(Y^n\oplus z^n)\mid X^n=x^n\oplus z^n}\\
     & = \EE{\sum_{\calS\subseteq [n]} \widehat{d}_\calS
         \chi_\calS(Y^n\oplus z^n)\mid X^n=x^n\oplus z^n}\\
     &=  \EE{\sum_{\calS\subseteq [n]} \widehat{d}_\calS
         \chi_\calS(Y^n)\chi_\calS(z^n)\mid X^n=x^n\oplus z^n}\\
         &\stackrel{(a)}{=}  \sum_{\calS\subseteq [n]} c_\calS \widehat{d}_\calS
         \chi_\calS(x^n\oplus z^n)\chi_\calS(z^n)\\
     &= \sum_{\calS\subseteq [n]} c_\calS \widehat{d}_\calS
         \chi_\calS(x^n)\\
         &\stackrel{(b)}{=} \EE{\sum_{\calS\subseteq [n]} \widehat{d}_\calS
         \chi_\calS(Y^n)|X^n=x^n}\\
         &= \EE{\delta_{y^n}(Y^n)\mid X^n=x^n}\\
         &=p_{Y^n|X^n}(y^n|x^n).
     \end{align*}
    Equalities $(a)$ and $(b)$ follow from the definition of $\calA_n$.
    By defining the distribution of $Z^n$ as
    $p_{Z^n}(z^n)\defined p_{Y^n|X^n}(z^n|\ones^n)$, where $\ones^n$ is the vector with
    all entries equal to 1, it follows that $Z^n=X^n\oplus Y^n$,
    $Z^n\independent X^n$ and
    $p_{Y^n|X^n}\subseteq \mathcal{B}_n$.

  \end{proof}

  The previous theorem suggests that there is a correspondence between the
  coefficients $c_\calS$ in \eqref{eq:PAchannels} and the distribution of the
  additive noise $Z^n$ in the definition of $\mathcal{B}_n$. The next result
  shows that this is indeed the case and, when $X^n$ is uniformly distributed,
  the coefficients $c_\calS^2$ correspond to the principal inertia components
  between $X^n$ and $Y^n$.

  \begin{thm}
    \label{thm:PIbinary}
    Let $p_{Y^n|X^n}\in \mathcal{B}_n$, and $X^n\sim p_{X^n}$. Then
    $\bP_{X^n,Y^n}=\bD_{X^n}\bbH_{2^n}\bLambda \bbH_{2^n}$, where $\bbH_l$ is
    the $l\times l$ normalized Hadamard matrix (i.e. $\bbH_l^2=\eye$).
    Furthermore, for $Z^n\sim p_{Z^n}$,  $\diag{\bLambda} =
    2^{n/2}\bbH_{2^n}\mathbf{p}_{Z^n}$, and the diagonal entries of $\bLambda$
    are equal to $c_\calS$ in \eqref{eq:PAchannels}. Finally, if $X$ is
    uniformly distributed, then $c_\calS^2$ are the principal inertia components
    of $X^n$ and $Y^n$.
  \end{thm}
  \begin{proof}
    Let  $p_{Y^n|X^n}\in \calA_n$ be given. From Theorem \ref{thm:AB} and the
    definition of $\calA_n$, it follows that $\chi_\calS(Y^n)$ is a right
    eigenvector of $p_{Y^n|X^n}$ with corresponding eigenvalue $c_\calS$.  Since
    $\chi_\calS(Y^n)2^{-n/2}$ corresponds to a row of $\bbH_{2^n}$ for each
    $\calS$ (due to the Kronecker product construction of the Hadamard matrix) and $\bbH_{2^n}^2=\eye$, then
    $\bP_{X^n,Y^n}=\bD_{X^n}\bbH_{2^n}\bLambda \bbH_{2^n}$. Finally, note that
    $\mathbf{p}_Z^T=2^{-n/2}\ones^T\bLambda\bbH_{2^n}$. From Lemma
    \ref{lem:conf}, it follows that  $c_\calS^2$ are the principal inertia
    components of $X^n$ and $Y^n$ if $X^n$ is uniformly distributed.
  \end{proof}

  \begin{remark}
    Theorem \ref{thm:PIbinary} indicates that one possible method for estimating the
      distribution of the additive binary noise $Z^n$ is to estimate  its effect 
      on the  parity bits of $X^n$ and $Y^n$. In this case, we are
      estimating the coefficients $c_\calS$ of the Walsh-Hadamard transform of
      $p_{Z^n}$. This approach was  studied by Raginsky \etal
       in \cite{raginsky_recursive_2013}.
  \end{remark}

Theorem \ref{thm:PIbinary} illustrates the filtering role of the principal
inertia components, discussed in the beginning of this section. If $X^n$ is
uniform, and using the same notation as in \eqref{eq:OBF}, then the vector of
conditional probabilities $\fb$ is transformed into the vector of \textit{a
posteriori} probabilities $\gb$ by: (i) taking the Hadamard transform of $\fb$,
(ii) filtering the transformed vector according to the coefficients $c_\calS$, where
$\calS\in [n]$, and (iii) taking the inverse Hadamard transform. The same
rationale applies to the transformation of $\Px$ into $\Py$ in binary additive
channels.

\subsection{Quantifying the Information of a Boolean Function of the Input of a
Noisy Channel}

We now investigate the connection between the principal inertia components
and $f$-information in the context of one-bit functions of $X$.
Recall from the discussion in the beginning of this section and, in particular,
equation \eqref{eq:OBF},  that for a binary $B$ and $B\rightarrow X \rightarrow
Y$, the distribution of $B$ and $Y$ is entirely specified by the transformation
of $\fb$ into $\gb$, where $\fb$ and $\gb$ are vectors with entries equal to
$p_{B|X}(0|\cdot)$ and $p_{B|Y}(0|\cdot)$, respectively.

For $\EE{B}=1-a$, the $f$-information between $B$ and $Y$ is given
by\footnote{Note that here we assume that $\calY=[n]$, so there is no ambiguity
in indexing $p_{B|Y}(0|Y)$ by $\gb_Y$.}
    \begin{align*}
            I_f(B;Y) =\ExpVal{}{ a f\left( \frac{\gb_Y}{a}
            \right) +(1-a) f\left( \frac{1-\gb_Y}{1-a}
            \right) }.
    \end{align*}

For $0\leq r,s \leq 1$, we can expand $f\left( \frac{r}{s} \right)$ around
1 as

    \begin{equation*}
        f\left( \frac{r}{s} \right) = \sum_{k=1}^\infty  \frac{
        f^{(k)}(1)}{k!}\left( \frac{r-s}{r} \right)^k.
    \end{equation*}
Denoting 
    \begin{align*}
        c_k(\alpha) &\defined
        \frac{1}{a^{k-1}}+\frac{(-1)^k}{(1-a)^{k-1}},
    \end{align*}
the $f$-information can then be expressed as
    \begin{align}
        I_f(B;Y)& = \sum_{k=2}^\infty  \frac{
         f^{(k)}(1)c_k(a)}{k!}\ExpVal{}{(\gb_Y-a)^k}.
     \label{eq:momentDecomp}
    \end{align}

    Similarly to \cite[Chapter 4]{_information_2004}, for a fixed $\EE{B}=1-a$, maximizing
the principal inertia components between $X$ and $Y$ will always maximize the
first term in the expansion \eqref{eq:momentDecomp}. To see why this is the
case, observe that
    \begin{align}
        \ExpVal{}{(\gb_Y-a)^k} &= (\gb-a)^T\bD_Y(\gb-a) \nonumber\\
                               &= \gb^T\bD_Y\gb -a^2 \nonumber\\
                               &=\fb^T\bD_X^{1/2}\bU\bSigma^2\bU^T\bD_x^{1/2}\fb-a^2.\label{eq:var}
     \end{align}
For a fixed $a$ and any $\fb$ such that $\fb^T\ones=a$, \eqref{eq:var} is
non-decreasing in the diagonal  entries of $\bSigma^2$ which, in turn, are
exactly the principal inertia components of $X$ and $Y$. Equivalently,
\eqref{eq:var} is non-decreasing in the $\chi^2$-divergence between $\pxy$ and
$\px\py$.

However, we do note that increasing the principal inertia components \textbf{does not}
increase the $f$-information between $B$ and $Y$ in general. Indeed, for a fixed
$\bU$, $\bV$ and marginal distributions of $X$ and $Y$, increasing the
principal inertia components might not even lead to a valid probability
distribution matrix $\bP$. 

Nevertheless, if $\bP$ is conforming and $X$ and $Y$ are uniformly distributed
over $[q]$, as shown in Lemma \eqref{lem:qary}, by increasing the principal
inertia components we can define a new random variable $\tilde{Y}$ that results
from sending $X$ through a $(\epsilon,q)$-SC, where $\epsilon$ is given in
\eqref{eq:epsilon}. In this case, the $f$-information between $B$ and $Y$ has a
simple expression when $B$ is a function of $X$.

\begin{lem}
  \label{lem:qaryIf}
  Let $B\rightarrow X \rightarrow \Yt$, where $B=h(X)$ for some $h:[q]\rightarrow
  \{0,1\}$,  $\EE{B}=1-a$ where $aq$ is an integer, $X$ is uniformly distributed in $[q]$ and $\Yt$ is the
  result of passing $X$ through a $(\epsilon,q)$-SC with $\epsilon\leq (q-1)/q$. Then
  \begin{equation}
    \label{eq:genaralqary}
    I_f(B;\Yt)=a^2f\left( 1+\sigma_1 c
    \right)+2a(1-a)f\left(1-\sigma_1\right)+(1-a)^2f\left( 1+\sigma_1c^{-1} \right)
  \end{equation}
  where $\sigma_1=\rho_m(X;\Yt)=1-\epsilon q (q-1)^{-1}$ and $c\defined (1-a)a^{-1}$. In particular,
  for $f(x)=x\log x$, then $I_f(X;\Yt)=I(X;\Yt)$, and for $\sigma_1 = 1-2\delta$
    \begin{align}
        I(B;\Yt)&=h_b(a)-\alpha H_b\left( 2\delta(1-a) \right)-(1-a)H_b(2\delta
        a) \label{eq:qaryMI}\\
        &\leq 1-H_b(\delta) \label{eq:qaryMax},
    \end{align}
  where $H_b(x)\defined-x\log(x)-(1-x)\log(1-x)$ is the binary entropy function.
\end{lem}

\begin{proof}
  Since $B$ is a deterministic function of $X$ and $aq$ is an integer, $\fb$ is
  a vector with $aq$ entries equal to 1 and $(1-a)q$ entries equal to 0. It
  follows from \eqref{eq:qarydist} that
  \begin{align*}
    I_f(B;\Yt)=&\frac{1}{q}\sum_{i=1}^q af\left(
    \frac{(1-\sigma_1)a+\fb_i\sigma_1}{a} \right)+(1-a)f\left(
    \frac{1-(1-\sigma_1)a-\fb_i\sigma_i}{1-a} \right)\\
    =&a^2f\left( 1+\sigma_1 \frac{1-a}{a}
    \right)+2a(1-a)f\left(1-\sigma_1\right)+(1-a)^2f\left(
    1+\sigma_1\frac{a}{1-a} \right).
  \end{align*}
  Letting $f(x)=x\log x$, \eqref{eq:qaryMI} follows immediately. Since
  \eqref{eq:qaryMI} is concave in $a$ and symmetric around $a=1/2$, it is
  maximized at $a=1/2$, resulting in \eqref{eq:qaryMax}.
\end{proof}

\subsection{On the ``Most Informative Bit'' Conjecture}
We now return to channels with additive binary noise, analyzed
is Section \ref{sec:binAdd}. Let $X^n$ be a uniformly distributed binary string
of length $n$ ($\calX = \{-1,1\})$,
and $Y^n$  the result of passing $X^n$ through a memoryless binary
symmetric channel with crossover probability $\delta\leq 1/2$. Kumar and Courtade
conjectured  \cite{kumar_which_2013} that for all binary $B$ and
$B\rightarrow X^n \rightarrow Y^n$ we have
    \begin{equation}
      I(B;Y^n)\leq 1-H_b(\delta).~~\mbox{(conjecture)} \label{eq:conjI}
    \end{equation}
It is sufficient to consider $B$ a function of $X^n$, denoted by $B=h(X^n)$,
$h:\{-1,1\}^n\rightarrow \{0,1\}$, and we make this
assumption henceforth. 

From the discussion in Section \ref{sec:binAdd}, for the memoryless binary
symmetric channel $Y^n=X^n\oplus Z^n$, where
$Z^n$ is an i.i.d. string with $\Pr\{Z_i=1\}=1-\delta$, and any
$\calS\in [n]$,
\begin{align*}
  \EE{\chi_\calS(Y^n)|X^n}&=
  \chi_\calS(X^n)\left(\Pr\left\{\chi_\calS(Z^n)=1\right\}-\Pr\left\{\chi_\calS(Z^n)=-1\right\}\right)\\
                          &=
                          \chi_\calS(X^n)\left(2\Pr\left\{\chi_\calS(Z^n)=1\right\}-1\right)\\
                          &=\chi_\calS(X^n)(1-2\delta)^{|\calS|}.
\end{align*}
It follows directly that $c_\calS = (1-2\delta)^{|\calS|}$ for all
$\calS\subseteq [n]$. Consequently, from Theorem \ref{thm:PIbinary}, the
principal inertia components of $X^n$ and $Y^n$ are of the form
$(1-2\delta)^{2|\calS|}$ for some $\calS\subseteq [n]$. Observe that the
principal inertia components act as a low pass filter on the vector of
conditional probabilities $\fb$ given in \eqref{eq:OBF}.

Can the noise distribution be modified so that the principal inertia components
act as an all-pass filter? More specifically, what happens when 
$\Yt^n=X^n\oplus W^n$, where $W^n$ is such that the principal inertia
components between $X^n$ and $\Yt^n$ satisfy $\sigma_i=1-2\delta$?
Then, from Lemma \ref{lem:qary}, $\Yt^n$ is the result of sending $X^n$ through
a $(\epsilon,2^n)$-SC with $\epsilon=2\delta(1-2^{-n})$. Therefore, from \eqref{eq:qaryMax}, 
    \begin{equation*}
        I(B;\Yt^n)\leq 1-H_b(\delta).
    \end{equation*}

    For any function $h:\{-1,1\}^n\rightarrow \{0,1\}$ such that $B=h(X^n)$, from standard results in Fourier
analysis of Boolean functions  \cite[Prop.  1.1]{odonnell_topics_2008}, $h(X^n)$
can be expanded as \[h(X^n)=\sum_{\calS\subseteq [n]}\hat{h}_\calS
\chi_\calS(X^n) .\] The value of $B$ is uniquely determined by the
action of $h$ on $\chi_\calS(X^n)$. Consequently, for a fixed function $h$, one
could expect that $\Yt^n$ should be more informative about $B$ than $Y^n$, since
the parity bits $\chi_\calS(X^n)$ are more reliably estimated from $\Yt^n$ than
from $Y^n$. Indeed, the memoryless
binary symmetric channel attenuates $\chi_\calS(X^n)$ exponentially  in
$|\calS|$, acting (as argued previously) as a low-pass filter. In addition, if
one could prove that for any fixed $h$ the inequality  $I(B;Y^n)\leq I(B;\Yt^n)$
holds, then \eqref{eq:conjI} would be proven true. This motivates the following
conjecture.


\begin{conj}
  \label{conj}
    For all $h:\{-1,1\}^n\rightarrow \{0,1\}$ and $B=h(X^n)$ 
    \begin{align*}
        I(B;Y^n)\leq I(B;\Yt^n).
    \end{align*}
\end{conj}

We note that Conjecture \ref{conj} is not true in general if $B$ is not a
deterministic function of $X^n$. In the next section, we provide further
evidence for this conjecture by investigating information metrics between $B$
and an estimate $\whB$ derived from $Y^n$.


\section{One-bit Estimators}
\label{sec:estimators}

Let $B\rightarrow X \rightarrow Y \rightarrow \widehat{B}$, where $B$ and $\whB$
are binary random variables with $\EE{B}=1-a$ and $\mathbb{E}[\whB]=1-b$. 
 We denote by $\xb\in \Reals^m$ and $\yb\in \Reals^n$ the
column vectors with entries $\xb_i = p_{B|X}(0|i)$ and $\yb_j =
p_{\whB|Y}(0|j)$. The joint distribution matrix of $B$ and $\whB$ is given by
\begin{equation}
  \label{eq:Pbbh}
  \mathbf{P}_{B,\whB}=\left(
\begin{array}{cc}
 z  & a-z  \\
 b-z  & 1-a-b+z \\
\end{array}
\right),
\end{equation}
where $z=\xb^T\bP\yb=\Pr\{B=\whB=0 \}$. For fixed values of $a$ and $b$, the joint distribution
of $B$ and $\whB$ only depends on $z$.

Let $f:\calP_{2\times 2}\rightarrow \Reals$, and, with a slight abuse of
notation, we also denote $f$ as a function of the entries of the $2\times 2$ matrix as
$f(a,b,z)$. If $f$ is convex in $z$ for a fixed $a$ and $b$, then $f$ is
maximized at one of the extreme values of $z$. Examples of such functions $f$
include mutual information and expected error probability.  Therefore,
characterizing the maximum and minimum values of $z$ is equivalent to
characterizing the maximum value of $f$ over all possible mappings $X\rightarrow
B$ and $Y\rightarrow \whB$. This leads to the following definition. 
\begin{defn}
    For a fixed $\bP$, the minimum and maximum
    values of $z$ over all possible mappings $X\rightarrow B$ and $Y\rightarrow
    \whB$ where  $\EE{B}=1-a$ and $\mathbb{E}[\whB]=1-b$ is defined as
    \begin{equation*}
      z^*_l(a,b,\bP) \defined \min_{\substack{\xb \in \calC^m(a,\bP^T)\\\yb \in
      \calC^n(b,\bP)}}
      \xb^T\bP\yb ~\mbox{ and }~
       z^*_u(a,b,\bP) \defined \max_{\substack{\xb \in \calC^m(a,\bP^T)\\\yb \in
       \calC^n(b,\bP)}}
       \xb^T\bP\yb,
    \end{equation*}
    respectively, and $\calC^{n}(a,\bP)$ is defined in \eqref{eq:Cdef}.
    
\end{defn}


The next lemma provides a simple upper-bound for $z^*_u(a,b,\bP)$ in terms of
the largest principal inertia components or, equivalently, the maximal
correlation between $X$ and $Y$.

\begin{lem}
  \label{lem:zupper}
  $z_u^*(a,b,\bP)\leq ab+\rho_m(X;Y)\sqrt{a(1-a)b(1-b)}$.
\end{lem}
\begin{remark}
  An analogous result was derived by Witsenhausen  \cite[Thm. 2]{witsenhausen_sequences_1975} for bounding the
    probability of agreement of a common bit derived from two correlated
    sources.
\end{remark}
\begin{proof}
    Let $\xb \in \calC^m(a,\bP^T)$ and $\yb \in \calC^n(b,\bP)$. Then, for $\bP$
     decomposed as in \eqref{eq:PIDecomp} and 
    $\bSigma^{-}=\diag{0,\sigma_1,\dots,\sigma_d}$,
    \begin{align} 
      \xb^T\bP\yb & = ab + \xb^T\bD_X^{1/2}\bU\bSigma^{-}\bV^T\bD_Y^{1/2}\yb
      \nonumber\\
      & = ab+\hat{\xb}^T\bSigma^{-}\hat{\yb}, \label{eq:zExpand}            
    \end{align}
    where $\hat{\xb}\defined \bU^T\bD_X^{1/2}\xb$ and $\hat{\yb}\defined
    \bV^T\bD_Y^{1/2}\yb$. Since  $\hat{\xb}_1=\normEuc{\hat{\xb}}=a$ and
    $\hat{\yb}_1=\normEuc{\hat{\yb}}=b$, then
    \begin{align*}
      \hat{\xb}^T\bSigma^{-}\hat{\yb} &= \sum_{i=2}^{d+1}\sigma_{i-1}
      \hat{\xb}_i\hat{\yb}_i\\
      &\leq
      \sigma_1\sqrt{\left(\normEuc{\hat{\xb}}^2-\hat{\xb}_1^2
      \right)\left(\normEuc{\hat{\yb}}^2-\hat{\yb}_1^2\right)}\\
      &= \sigma_1\sqrt{(a-a^2)(b-b^2)}.
    \end{align*}
    The result  follows by noting that $\sigma_1=\rho_m(X;Y)$.
\end{proof}

We will focus in the rest of this
section on functions and corresponding estimators that are (i) unbiased ($a=b$) and
(ii) satisfy $z=\Pr\{\hat{B}=B=0\}\geq a^2$. The set of all such mappings is
given by
\begin{equation*}  
    \mathcal{ H }(a,\bP)\defined \left\{ (\bx,\by)\mid \bx\in
    \calC^m(a,\bP^T),\by\in \calC^n(a,\bP),\bx^T\bP \by\geq a^2 \right\}.
\end{equation*}

The next results provide upper and lower bounds for $z$ for the mappings in
$\calH(a,\bP)$.

\begin{lem}
  \label{lem:zbounds}
    Let $0\leq a\leq 1/2$ and $\bP$ be fixed. For any $(\bx,\by)\in  \calH(a,\bP)$
        \begin{equation}
          \label{eq:zbounds}
            a^2\leq z \leq a^2+\rho_m(X;Y)a(1-a),
        \end{equation}  
        where  $z=\bx^T\bP\by$.
\end{lem}
\begin{proof}
    The lower bound for $z$ follows directly from the definition of
    $\calH(a,\bP)$, and the upper bound follows from Lemma \ref{lem:zupper}. 
\end{proof}

The previous lemma allows us to provide an upper bound over the mappings in
$\calH(a,\bP)$ for the $f$-information between $B$ and $\whB$ when $I_f$ is
non-negative. 

\begin{thm}
  \label{thm:Estimators}
  For any non-negative $I_f$  and fixed $a$ and $\bP$,
  \begin{equation} 
    \label{eq:unbiasedBound}
    \sup_{ (\bx,\by) \in \calH(a,\bP)} I_f(B;\hat{B})\leq  a^2f\left( 1+\sigma_1 c
    \right)+2a(1-a)f\left(1-\sigma_1\right)+(1-a)^2f\left( 1+\sigma_1c^{-1} \right)
  \end{equation}
  where here $\sigma_1=\rho_m(X;\Yt)$ and $c\defined (1-a)a^{-1}$. In particular, for
  $a=1/2$,
    \begin{equation} 
    \sup_{ (\bx,\by) \in \calH(1/2,\bP)} I_f(B;\hat{B})\leq
  \frac{1}{2}\left( f(1-\sigma_1)+f(1+\sigma_1) \right). 
  \end{equation}
\end{thm}
\begin{proof}
     Using the matrix form of the joint distribution between $B$ and $\whB$
     given in \eqref{eq:Pbbh}, for $\EE{B}=\EE{\whB}=1-a$, the $f$ information is given by
    \begin{align}
      \label{eq:Ifproof}
      I_f(B;\hat{B}) = a^2f\left( \frac{z}{a^2} \right)+ 2a(1-a)f\left(
      \frac{a-z}{a(1-a)} \right)
                        + (1-a)^2f\left( \frac{1-2a+z}{(1-a)^2} \right).
    \end{align}
    Consequently, \eqref{eq:Ifproof} is convex in $z$. For $(\bx,\by)
    \in \calH(a,\bP)$, it follows from Lemma \ref{lem:zbounds} that $z$ is restricted to the
    interval in \eqref{eq:zbounds}. Since  $I_f(B;\hat{B})$ is non-negative by
    assumption, $I_f(B;\hat{B})=0$ for $z=a^2$ and
    \eqref{eq:Ifproof} is convex in $z$, then $I_f(B;\hat{B})$  is non-decreasing in $z$ for $z$ in
    \eqref{eq:zbounds}. Substituting $z=a^2+\rho_m(X;Y)a(1-a)$ in
    \eqref{eq:Ifproof}, inequality \eqref{eq:unbiasedBound} follows.    
  \end{proof}

\begin{remark}
  Note that the right-hand side of \eqref{eq:unbiasedBound} matches  the
  right-hand side of \eqref{eq:genaralqary}, and provides further evidence for
  Conjecture \ref{conj}. This result indicates that, for
  conforming probability distributions, the information between a binary
  function and its corresponding unbiased estimate is maximized when all the principal
  inertia components have the same value.
\end{remark}

Following the same approach from Lemma \ref{lem:qaryIf}, we find the next 
bound for the mutual information between $B$ and $\whB$. 
\begin{cor}
  \label{cor:MIEstimators}
  For $a$ fixed and $\rho_m(X;Y)=1-2\delta$ 
  \begin{align*}
    \sup_{ (\bx,\by) \in \calH(a,\bP)} I(B;\hat{B})\leq
    1-H_b(\delta). 
  \end{align*}
\end{cor}

We now provide a few application examples for the results derived in this section.

\subsection{Lower Bounding the Estimation Error Probability}

For $z$ given in \eqref{eq:Pbbh}, the average estimation error probability is
given by $\Pr\{B\neq
\whB\}=a+b-2z$, which is a convex (linear) function of $z$. If  $a$ and $b$ are fixed,
then the error probability  is minimized when $z$ is
maximized. Therefore
\begin{equation*}
   \Pr\{B\neq \whB\}\geq a+b-2z_u^*(a,b). 
\end{equation*}
Using the bound from Lemma \ref{lem:zupper}, it follows that
\begin{equation}
  \label{eq:errorProb}
  \Pr\{B\neq \whB\}\geq a+b-2ab-2\rho_m(X;Y)\sqrt{a(1-a)b(1-b)}. 
\end{equation}
The bound \eqref{eq:errorProb} is exactly the bound derived by Witsenhausen in
\cite[Thm 2.]{witsenhausen_sequences_1975}. Furthermore, minimizing the right-hand side of  \eqref{eq:errorProb}
over $0\leq b \leq 1/2$, we arrive at
\begin{equation}
  \label{eq:Fano}
      \Pr\{B\neq \whB\}\geq\frac{1}{2}\left(
  1-\sqrt{1-4a(1-a)(1-\rho_m(X;Y)^2)} \right), 
  \end{equation}
  which is a particular form of the bound derived by Calmon \etal \cite[Thm. 3]{calmon_bounds_2013}.

\subsection{Memoryless Binary Symmetric Channels with Uniform Inputs}

We now turn our attention back to the setting considered in Section
\ref{sec:binAdd}. Let $Y^n$ be the result of passing $X^n$ through a memoryless binary symmetric channel
with crossover probability $\delta$, $X^n$ uniformly distributed, and
$B\rightarrow X^n\rightarrow Y^n\rightarrow \whB$. Then
$\rho_m(X^n;Y^n)=1-2\delta$ and, from
\eqref{eq:Fano}, when $\EE{B}=1/2$, 
\begin{equation*}
     \Pr\{B\neq \whB\}\geq \delta.
\end{equation*}
Consequently, inferring any unbiased one-bit function of the input of a binary
symmetric channel is at least as hard (in terms of error probability) as
inferring a single output from a single input. 

Using the result from Corollary \ref{cor:MIEstimators}, it follows that when
$\EE{B}=\EE{\whB}=a$ and $\Pr\{B=\whB=0\}\geq a^2 $, then 
\begin{equation}
  \label{eq:Sudeep}
    I(B;\whB)\leq 1-H_b(\delta).
\end{equation}
\begin{remark}
    Anantharam \etal presented in \cite{anantharam_hypercontractivity_2013} a computer aided proof that the upper
    bound \eqref{eq:Sudeep} holds  for any $B\rightarrow
    X^n\rightarrow Y^n\rightarrow \whB$. However, we highlight that the methods
    introduced here allowed an analytical derivation of the inequality
    \eqref{eq:Sudeep}, which, in turn, is a particular case of the more general
    setting studied by Anantharam \etal
 \end{remark}

\section*{Acknowledgement}
The authors would like to thank  Prof. Yury Polyanskiy for the insightful
discussions and suggestions throughout the course of this work.

\bibliography{references}
\bibliographystyle{IEEEtran}

\end{document}